\documentclass[journal]{IEEEtran} 				% For final IEEE paper version
\IEEEoverridecommandlockouts 						% For final IEEE paper version
\usepackage{amsmath,amssymb}
\usepackage{amsthm}
\usepackage{color}
\usepackage[mathscr]{eucal}
\usepackage{graphics,graphicx,multicol}
\usepackage{epstopdf}
\usepackage{enumerate}
\usepackage{subfigure}
\usepackage{algorithmic}
\usepackage{algorithm}
\usepackage{morefloats}
\usepackage{enumerate}
\usepackage{subfigure}
\usepackage{multirow}
\usepackage{array}
\usepackage{pstricks, pst-node, pst-plot, pst-circ}
\usepackage{moredefs}
\usepackage{cite}
\usepackage{courier}
\newtheorem{thm}{Theorem}[section]
\newtheorem{lem}[thm]{Lemma}

\begin{document}
\title{An Efficient Multi-Carrier Resource Allocation with User Discrimination Framework for 5G Wireless Systems}
\author{Haya~Shajaiah,
        Ahmed~Abdelhadi
        and~Charles~Clancy
\thanks{H. Shajaiah, A. Abdelhadi, and C. Clancy are with the Hume Center for National Security and Technology, Virginia Tech, Arlington,
VA, 22203 USA e-mail: \{hayajs, aabdelhadi, tcc\}@vt.edu.}
%\thanks{Manuscript received April 19, 2014; revised May 27, 2014.}
}
\maketitle
\begin{abstract}
In this paper, we present an efficient resource allocation with user discrimination framework for 5G Wireless Systems to allocate multiple carriers resources among users with elastic and inelastic traffic. Each application running on the user equipment (UE) is assigned an application utility function. In the proposed model, different classes of user groups are considered and users are partitioned into different groups based on the carriers coverage area. Each user has a minimum required application rate based on its class and the type of its application. Our objective is to allocate multiple carriers resources optimally among users, that belong to different classes, located within the carriers' coverage area. We use a utility proportional fairness approach in the utility percentage of the application running on the UE. Each user is guaranteed a minimum quality of service (QoS) with a priority criterion that is based on user's class and the type of application running on the UE. In addition, we prove the existence of optimal solutions for the proposed resource allocation optimization problem and present a multi-carrier resource allocation with user discrimination algorithm. Finally, we present simulation results for the performance of the proposed algorithm.
\end{abstract}


% Global Parameters that can be changed:
%%%%%%%%%%%%%%%%%%%%%%%%%%%%%%%%%%%
\providelength{\AxesLineWidth}       \setlength{\AxesLineWidth}{0.5pt}%
\providelength{\plotwidth}           \setlength{\plotwidth}{8cm}% width of the axes only
\providelength{\LineWidth}           \setlength{\LineWidth}{0.7pt}%
\providelength{\MarkerSize}          \setlength{\MarkerSize}{3pt}%
\newrgbcolor{GridColor}{0.8 0.8 0.8}%
\newrgbcolor{GridColor2}{0.5 0.5 0.5}%
%%%%%%%%%%%%%%%%%%%%%%%%%%%%%%%%%%%
%%%%%%%%%%%%%%%%%%%%%%%%%%%%%%%%%%%

\section{Introduction}\label{sec:intro}
Recently, there has been a massive growth in the number of mobile users and their traffic. The data traffic volume almost doubles every year. Mobile users are currently running multiple applications that require higher bandwidth which makes users so limited to the service providers' resources. Multiple services  are now offered by network providers such as mobile-TV and multimedia telephony \cite{QoS_3GPP}. According to the Cisco Visual Networking Index (VNI) \cite{Visual-Networking}, the volume of data traffic is expected to continue growing up and reaches 1000 times its value in 2010 by 2020 which is referred to as 1000x data challenge. With the increasing volume of data traffic, more spectrum is required \cite{Carrier_Agg_1}. However, due to spectrum scarcity, it is difficult to provide the required resources with a single frequency band. Therefore, aggregating frequency bands, that belong to different carriers, is needed to utilize the radio resources across multiple carriers and expand the effective bandwidth delivered to user terminals, leading to interband non-contiguous carrier aggregation \cite{Carrier_Agg_2}.

As the fourth generation long term evolution (4G-LTE) system is now reaching maturity and only small amounts of new spectrum is expected, researchers have started to establish the foundation of 5G system that should be coming next. The expected capabilities of 5G systems have started to take shape. The 5G networks have promised to handle multiple applications with various QoS requirements, address the 1000x data challenge and provide as low as 1 millisecond latency. Due to the new emerging applications that are beyond personal communications, the number of wireless devices could reach hundreds of billions by the time 5G comes to fruition \cite{The-challenges,proximity-aware}. Because of the urgently need for 5G, government and industries agencies have initiated the research and development process of 5G. The development of 5G requires efforts in three areas: enhancement in spectrum efficiency, spectrum expansion and dense deployment of small cells.

Carrier aggregation (CA) is one of the most distinct features of 4G systems including Long Term Evolution Advanced (LTE Advanced). As 5G systems' expected capabilities have started to take shape, CA is expected to be supported by 5G. Therefore, CA needs to be taken into consideration when designing 5G systems. With CA being applied, wider transmission bandwidths between the evolve node B (eNodeB) and the UE can be achieved by aggregating multiple component carriers (CCs) of the same or different bandwidths. An overview of CA framework and cases is presented in \cite{CA-framework,work-item}. %In order to improve the system capacity and performance, many operators are willing to add the CA feature to their plans.
Beside CA capability, 5G wireless network promises to handle diverse QoS requirements of multiple applications since different applications require different application's performance. Furthermore, certain types of users may require to be given priority when allocating the network resources (i.e. such as public safety users) which needs to be taken into consideration when designing the resource allocation framework.
%
%removed because introduction is lengthy
%Utilizing the existing spectrum more efficiently can significantly improve network capacity, data rates and user experience. Some spectrum holders such as government users do not use their entire allocated spectrum in every part of their geographic boundaries most of the time. Therefore, the National Broadband Plan (NBP) and the findings of the President's Council of Advisors on Science and Technology (PCAST) spectrum study have recommended making the under-utilized federal spectrum available for secondary use \cite{PCAST}. Using spectrum sharing, wireless systems will be able to harvest under-utilized swathes of spectrum, which would improve the efficiency of spectrum usage. With more spectrum available, significant gain in mobile broadband capacity can be achieved if those resources are aggregated efficiently with the existing commercial wireless system resources.

A multi-stage resource allocation (RA) with carrier aggregation algorithms are presented in \cite{Haya_Utility1,Haya_Utility3,Haya_Utility6}. The RA with CA algorithm in \cite{Haya_Utility1} uses utility proportional fairness approach to allocate primary and secondary carriers resources optimally among users in their coverage area. The primary carrier first allocates its resources optimally among users in its coverage area. The secondary carrier then starts allocating its resources optimally to users in its coverage area based on the rates allocated to the users by the primary carrier and the users applications. A resource allocation with CA optimization problem is presented in \cite{Haya_Utility3} to allocate the LTE Advanced carrier and the MIMO radar carrier resources to each UE in a LTE Advanced cell based on the UE's applications. A price selective centralized resource allocation with CA algorithm is presented in \cite{Haya_Utility6} to allocate multiple carriers resources optimally among users while giving the user the ability to select its primary and secondary carriers. The carrier selection decision is based on the carrier price per unit bandwidth. A RA with user discrimination optimization framework is presented in \cite{Haya_Utility2} and \cite{Haya_Utility4} to allocate one carrier resources among users under the carrier's coverage area.

In this paper, we provide an efficient framework for the resource allocation problem to allocate multi-carrier resources optimally among users that belong to different classes of user groups. In our model, we use utility functions to represent users' applications. Sigmoidal-like utility functions and logarithmic utility functions are used to represent real-time and delay-tolerant applications, respectively, running on the UEs \cite{Ahmed_Utility1}.
The resource allocation with user discrimination framework presented in \cite{Haya_Utility4} does not consider the case of multi-carrier resources available at the eNodeB. It only solves the problem of resource allocation with user discrimination in the case of single carrier. In this paper, we consider the case of multiple carriers' resources available at the eNodeB and multiple classes of users located under the coverage area of these carriers. We use a priority criterion for the resource allocation process that varies based on the user's class and the type of application running on the UE. We consider two classes of users, VIP users (i.e. public safety users or users who require emergency services) and regular users. VIP users are assigned a minimum required application rate for each of their applications whereas regular users' applications are not assigned any.

We formulate the resource allocation with user discrimination problem in a multi-stage resource allocation with carrier aggregation optimization problem to allocate resources to each user from its all in range carriers based on a utility proportional fairness policy. Each application running on the UE is assigned an application minimum required rate by the network that varies based on the type of user's application and the user's class. Furthermore, if the user's in range carriers have enough available resources, the user is allocated at minimum its applications' minimum required rates. VIP users are given priority over regular users by the network when allocating each carrier's resources, and real-time applications are given priority over delay-tolerant applications.
\subsection{Related Work}\label{sec:related}

There has been several works in the area of optimizing resource allocation to achieve an efficient utilization of the scarce radio spectrum. In \cite{kelly98ratecontrol,Internet_Congestion,Optimization_flow,Fair_endtoend}, the authors have used utility functions to represent users traffic. They used a strictly concave utility function to represent elastic traffic and proposed distributed algorithms at the sources and the links to interpret the congestion control of communication networks. Their suggested approach only focussed on elastic traffic and did not consider real-time applications as it have non-concave utility functions as shown in \cite{fundamental_design}. In \cite{Utility_max-min} and \cite{ Fair_allocation}, the authors have argued that the utility function is the one that needs to be shared fairly, rather than the bandwidth, as it represents the performance of the user's application.% In this paper, we consider using resource allocation in order to maximize the user satisfaction and achieve utility proportional fairness. In the case of applying bandwidth proportional fairness through a max-min bandwidth allocation, the utilities received by delay-tolerant applications are larger than the utilities received by real-time applications as real-time applications require minimum encoding rates and their utilities are equal to zero if they do not receive the applications minimum encoding rates.

In \cite{kelly98ratecontrol}, the authors have introduced a proportional fairness resource allocation approach. However, their approach does not guarantee a minimum QoS for each user application. To overcome this issue, the authors in \cite{Ahmed_Utility1} introduced a utility proportional fairness resource allocation algorithm. Their approach respects the real-time applications inelastic behavior and therefore we believe that it is more appropriate. The utility proportional fairness resource allocation algorithm presented in \cite{Ahmed_Utility1} guarantees that no user is allocated zero rate and gives real-time applications priority over delay tolerant applications when allocating resources. In \cite{Ahmed_Utility1, Ahmed_Utility2} and \cite{ Ahmed_Utility3}, the authors have presented RA algorithms to allocate single carrier resources optimally among mobile users who are treated evenly. However, these algorithms do not support resource allocation with user discrimination for multi-carrier systems. To incorporate the carrier aggregation feature and the case of different classes of users, we have introduced a multi-stage resource allocation using carrier aggregation in \cite{Haya_Utility1}. Furthermore, in \cite{Haya_Utility2} and \cite{Haya_Utility4}, we presented resource allocation with users discrimination algorithms to allocate a single carrier resources optimally among mobile users running elastic and inelastic traffic. In \cite{Mo_ResourceBlock}, the authors have presented a radio resource block allocation optimization problem using a utility proportional fairness approach. The authors in \cite{Tugba_ApplicationAware} have presented an application-aware resource block scheduling approach for elastic and inelastic traffic by assigning users to resource blocks.

On the other hand, an extensive attention has been recently given to the resource allocation for single cell multi-carrier systems \cite{Dual-Decomposition, Resource_allocation, Rate_Balancing}. In \cite{Fair_resource,Design_of_Fair,Fast_Algorithms,Optimal_and_near-optimal}, the authors have represented this challenge in optimization problems frameworks. Their objective is to maximize the overall cell throughput while taking into consideration some constraints such as transmission power and fairness. However, rather than achieving better system-centric throughput, better user satisfaction can be achieved by transforming the problem into a utility maximization framework. The authors in \cite{Downlink_dynamic,Centralized_vs_Distributed} have focussed on reducing the implementation complexity and suggested using a distributed resource allocation rather than a centralized one. The authors in \cite{Cooperative_Fair_Scheduling} have proposed a collaborative scheme in a multiple base stations (BSs) environment, where each user is served by the BS with the best channel gain. The authors in \cite{DownlinkRadio} have addressed the problem of spectrum resource allocation with CA based LTE Advanced systems, by considering the UE's MIMO capability and the modulation and coding schemes (MCSs) selection.
\subsection{Our Contributions}\label{sec:contributions}
Our contributions in this paper are summarized as:
\begin{itemize}
\item We present a multi-stage resource allocation with user discrimination optimization problem to allocate multi-carrier resources optimally among different classes of users.
\item We prove that the resource allocation optimization problem is convex and therefore the global optimal solution is tractable.
\item We present a resource allocation algorithm to solve the optimization problem and allocate each user an aggregated final rate from its in range carriers. The proposed algorithm outperforms that presented in \cite{Haya_Utility4} as it considers allocating each user resources from multiple carriers using a resource allocation with carrier aggregation approach.
\item We present simulation results for the performance of the proposed resource allocation algorithm.
\end{itemize}

The remainder of this paper is organized as follows. Section \ref{sec:Problem_formulation} presents the problem formulation. In section \ref{sec:ResourceAllocation}, we present the resource allocation optimization problems for three cases and prove that the global optimal solution exists and is tractable. Section \ref{sec:Algorithm} presents our multi-carrier resource allocation with user discrimination algorithm. In section \ref{sec:sim}, we discuss simulation setup and provide quantitative results along with discussion. Section \ref{sec:conclude} concludes the paper.
%%%%%%%%%%%%%%%%%%%%%%%%%%%%%%%%%%%
%%%%%%%%%%%%%%%%%%%%%%%%%%%%%%%%%%%
%%%%%%%%%%%%%%%%%%%%%%%%%%%%%%%%%%%
\section{Problem Formulation}\label{sec:Problem_formulation}
In this paper, we consider a single cell mobile system with one eNodeB, $K$ carriers (frequency bands) that have resources available at the eNodeB, $M$ regular and VIP UEs. Let $\mathcal{M}$ be the set of all regular and VIP UEs where $M=|\mathcal{M}|$. The set of carriers is given by $\mathcal{K}=\{1,2,...,K\}$ with carriers in order from the highest frequency to the lowest frequency. Higher frequency carriers have smaller coverage area than lower frequency carriers. The eNodeB allocates resources from multiple carriers to each UE. Users located under the coverage area of multiple carriers are allocated resources from all in range carriers.
The rate allocated by the eNodeB to UE $i$ from all in range carriers is given by $r_i$. Each application running on the UE is mathematically represented by a utility function $U_i(r_i)$ that corresponds to the application's type and represents the user satisfaction with its allocated rate $r_i$ as shown in section \ref{sec:utility functions}. Our goal is to determine the optimal rates that the eNodeB shall allocate from each carrier to each UE in order to maximize the total system utility while ensuring proportional fairness between utilities.

The rate allocated to the $i^{th}$ user in $\mathcal{M}$ by the $j^{th}$ carrier in $\mathcal{K}$ is given by $r_i^{j,all}$. The final allocated rate by the eNodeB to the $i^{th}$ user is given by
\begin{equation}\label{eqn:rate r_i}
r_i = \sum_{j\in \mathcal{K}}r_i^{j,all}
\end{equation}
where $r_i$ is equivalent to the sum of rates allocated to the $i^{th}$ user from all carriers in its range. Based on the coverage area of each carrier and the users' classes, a user grouping method is introduced in \ref{sec:UsersGrouping} to partition users into groups. The eNodeB performs resource allocation with user discrimination based on carrier aggregation to allocate each carrier's resources to users located within the coverage area of that carrier.
%%%%%%%%%%%%%%%%%%%%%%%%%%%%%%%%%%%%%%%%%%%%%%%%%%%%%%%%%%%%%%%%%%%%%%%%%%%55
\subsection{Application Utility Functions}\label{sec:utility functions}
We express the user satisfaction with its rate using utility functions that represent the degree of satisfaction of the user function with the rate allocated by the cellular network \cite{DL_PowerAllocation,Fundamental,Utility-proportional,Ahmed_Utility2}. We represent the $i^{th}$ user application utility function $U_i(r_i)$ by sigmoidal-like function or logarithmic function where $r_i$ is the rate of the $i^{th}$ user. These utility functions have the following properties:
\begin{itemize}
\item $U_i(0) = 0$ and $U_i(r_i)$ is an increasing function of $r_i$.
\item $U_i(r_i)$ is twice continuously differentiable in $r_i$ and bounded above.
\end{itemize}

In our model, we use the normalized sigmoidal-like utility function, as in \cite{Ahmed_Utility2}, that can be expressed as
\begin{equation}\label{eqn:sigmoid}
U_i(r_i) = c_i\Big(\frac{1}{1+e^{-a_i(r_i-b_i)}}-d_i\Big),
\end{equation}
where $c_i = \frac{1+e^{a_ib_i}}{e^{a_ib_i}}$ and $d_i = \frac{1}{1+e^{a_ib_i}}$ so it satisfies $U_i(0)=0$ and $U_i(\infty)=1$. The normalized sigmoidal-like function has an inflection point at $r_i^{\text{inf}}=b_i$. In addition, we use the normalized logarithmic utility function, used in \cite{Ahmed_Utility1}, that can be expressed as
%The normalized sigmoidal-like utility functions with $a=5$ and $b=10$, and $a=5$ and $b=10$ is shown in Figure \ref{fig:sigmoid}.
\begin{equation}\label{eqn:log}
U_i(r_i) = \frac{\log(1+k_ir_i)}{\log(1+k_ir_i^{\text{max}})},
\end{equation}
where $r_i^{\text{max}}$ gives $100\%$ utilization and $k_i$ is the slope of the curve that varies based on the user application. So, it satisfies $U_i(0)=0$ and $U_i(r_i^{\text{max}})=1$.
%%%%%%%%%%%%%%%%%%%%%%%%%%%%%%%%%%%%%%%%
\subsection{User Grouping Method}\label{sec:UsersGrouping}
In this section we introduce a user grouping method to create user groups for each carrier $j \in \mathcal{K}$. The eNodeB creates a user group $\mathcal{M}_j$ for each carrier where $\mathcal{M}_j$ is a set of users located under the coverage area of the $j^{th}$ carrier. The number of users in $\mathcal{M}_j$ is given by $M_j=|\mathcal{M}_j|$. Furthermore, users in $\mathcal{M}_j$ are partitioned into two groups of users. A VIP user group $\mathcal{M}_j^{\text{VIP}}$ and a regular user group $\mathcal{M}_j^{\text{Reg}}$, where $\mathcal{M}_j^{\text{VIP}}$ and $\mathcal{M}_j^{\text{Reg}}$ are the sets of all VIP users and regular users, respectively, located under the coverage area of the $j^{th}$ carrier with $\mathcal{M}_j=\mathcal{M}_j^{\text{VIP}}\cup \mathcal{M}_j^{\text{Reg}}$. The number of users in $\mathcal{M}_j^{\text{VIP}}$ and $\mathcal{M}_j^{\text{Reg}}$ is given by $M_j^{\text{VIP}}=|\mathcal{M}_j^{\text{VIP}}|$ and $M_j^{\text{Reg}}=|\mathcal{M}_j^{\text{Reg}}|$, respectively. The eNodeB allocates the $j^{th}$ carrier resources to users in $\mathcal{M}_j$ with a priority given to VIP users (i.e. users in $\mathcal{M}_j^{\text{VIP}}$). Users located under the coverage area of multiple carriers (i.e. common users in multiple user groups) are allocated resources from these carriers and their final rates are aggregated under a non adjacent inter band aggregation scenario.

The $i^{th}$ user is considered part of user group $\mathcal{M}_j$ if it is located within a distance of $D_j$ from the eNodeB where $D_j$ represents the coverage radius of the $j^{th}$ carrier. Let $d_i$ denotes the distance between the eNodeB and user $i$. The $j^{th}$ carrier user group $\mathcal{M}_j$ is defined as
\begin{equation}\label{eqn:user_group}
\mathcal{M}_j = \{i : d_i<D_j, 1\leq i \leq M\}, 1 \leq j \leq K.
\end{equation}

On the other hand, the eNodeB creates a set of carriers $\mathcal{K}_i$, for each user, that is defined as
\begin{equation}\label{eqn:K_i}
\mathcal{K}_i = \{j : d_i<D_j, 1\leq j \leq K\}, 1 \leq i \leq M.
\end{equation}

The number of carriers that the $i^{th}$ user can be allocated resources from is given by $N_i=|\mathcal{K}_i|$. Higher frequency carriers have smaller coverage radius than lower frequency carriers (i.e. $D_1<D_2<...<D_K$). Therefore, user group $\mathcal{M}_1 \subseteq \mathcal{M}_2 \subseteq ...\subseteq \mathcal{M}_K$. Figure \ref{fig:UserGroups} shows one cellular cell with one eNodeB under non adjacent inter band scenario with $K$ carriers in $\mathcal{K}$ and $M$ users in $\mathcal{M}$ and how users are partitioned into user groups based on their location and their class.
\begin{figure}[tb]
\centering
\includegraphics[height=2.4in, width=2.4in]{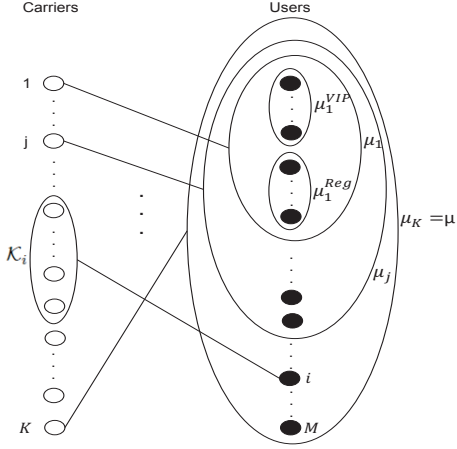}
\caption{User grouping for a LTE mobile system with $M$ users in $\mathcal{M}$ and $K$ carriers in $\mathcal{K}$. $\mathcal{M}_j$ represents the set of users located under the coverage area of the $j^{th}$ carrier with $\mathcal{M}_j=\mathcal{M}_j^{VIP}\cup \mathcal{M}_j^{Reg}$. $\mathcal{K}_i$ represents the set of all in range carriers for the $i^{th}$ user.}
\label{fig:UserGroups}
\end{figure}
%%%%%%%%%%%%%%%%%%%%%%%%%%%%%%%%%%%%%%%%
\section{Multi-Carrier Resource Allocation with User discrimination Optimization Problem}\label{sec:ResourceAllocation}
In this section, we present a multi-stage resource allocation (RA) with user discrimination optimization problem to allocate multi-carrier resources optimally among users in their coverage area. Our objective is to find the final allocated rate to each user from its all in range carriers based on a utility proportional fairness policy. We use utility functions of users rates to represent the type of application running on the UE. Every user subscribing for a mobile service is guaranteed to achieve a minimum QoS with priority criterion. VIP users are given priority when allocating each carrier's resources and within each user class group, whether it is VIP or regular user group, real time applications are given priority when allocating each carrier's resources. This is due to the nature of sigmoidal-like utility functions that are used to represent real-time applications.

The eNodeB performs the resource allocation process for all carriers one at a time and one after another in ascending order of their coverage radius $D_j$. Each carrier $j \in \mathcal{K}$ has a limited amount of available resources that is given by $R_j$ and each user's application has a minimum required rate $r_i^{\text{req}}$ that is equivalent to zero in the case of regular users and is equivalent to certain value (i.e. rate) in the case of VIP users. The eNodeB starts the RA process by performing a RA for carrier $1$ in $\mathcal{K}$ as it has the smallest coverage radius $D_1$. After allocating its resources to users in $\mathcal{M}_1$, the eNodeB then starts the RA process to allocate carrier $2$ resources to users in $\mathcal{M}_2$. In addition, since $\mathcal{M}_1 \subseteq \mathcal{M}_2$ the eNodeB allocates users in $\mathcal{M}_1$ resources from carrier $2$ and the rates are aggregated based on a non adjacent inter band aggregation scenario. The eNodeB continues the resource allocation process by allocating the $j^{th}$ carrier resources to users in $\mathcal{M}_j$. Let $r_i^{j,\text{all}}$ represents the rate allocated by the $j^{th}$ carrier to UE $i$ and let $C_i$ represents the total aggregated rate allocated to UE $i$ by carriers $\{1,2,...,j-1\}$ where $C_i=\sum_{l=1}^{j-1}r_i^{l,\text{all}}$. Furthermore, let $C_i^j$ be a constant that is always equivalent to zero for regular users whereas for VIP users $C_i^j$ is equivalent to zero or $r_i^{\text{req}}-C_i$ based on some conditions that are discussed later in this section. The resource allocation process is finalized by allocating the $K^{th}$ carrier resources to users in $\mathcal{M}_K$, i.e. all users in the cellular cell as they are all located within its coverage radius. We consider a utility proportional fairness objective function, based on carrier aggregation, that the eNodeB seeks to maximize for each time it allocates a carrier's resources.

The proposed RA optimization problem for multi-carrier cellular systems is divided into three cases. In order for the eNodeB to guarantee that VIP users are given priority when allocating each carrier's resources, each time the eNodeB performs a RA process for a carrier it checks the values of 1) the carrier's available resources $R_j$, 2) the current total rate allocated to each VIP UE $i \in \mathcal{M}_j^{\text{VIP}}$ from other carriers (i.e. $C_i=\sum_{l=1}^{j-1}r_i^{l,\text{all}}$) and 3) the value of $r_i^{\text{req}}-C_i$ for each VIP UE $i \in \mathcal{M}_j^{\text{VIP}}$ if $C_i < r_i^{\text{req}}$. Based on these values, the eNodB performs the RA process that corresponds to the most appropriate case among the three cases. The three cases and their RA optimization framework are presented below.\\
%%%%%%%%%%%%%%%%%%%%
\\\textit{Case$\;1.$ RA Optimization Problem when $C_i \geq r_i^{\text{req}}\:\:\forall i \in \mathcal{M}_j$}:

The eNodeB chooses the RA optimization problem of this case in order to allocate the $j^{th}$ carrier resources if the total aggregated rate $C_i$ that is allocated to each UE $i \in \mathcal{M}_j$ from carriers $\{1,2,...,j-1\}$ is greater than or equal the minimum required application rate $r_i^{\text{req}}$. In this case, since each UE has already been allocated at least its application minimum required rate from other carriers, the eNodeB performs the RA process among all users under the coverage area of carrier $j$. The RA optimization problem for the $j^{th}$ carrier in this case is given by:
\begin{equation}\label{eqn:opt_case1}
\begin{aligned}
& \underset{\textbf{r}^j}{\text{max}}
& & \prod_{i=1}^{M_j}U_i(C_i+C_i^j+r_i^j) \\
& \text{subject to}
& & \sum_{i=1}^{M_j}r_i^{j,\text{all}} \leq R_j, \;\; r_i^{j,\text{all}} \geq 0\\
& & & r_i^{j,\text{all}}=r_i^j+C_i^j, \;\; C_i^j=0\\
& & & C_i=\sum_{l=1}^{j-1}r_i^{l,\text{all}}, \;\; C_i \geq r_i^{\text{req}}, \;\;\; i = 1,2, ...,M_j,\\
\end{aligned}
\end{equation}
%%%%%%%%%%%%%%%%%%%%%%%%%%%%%%%%%%%%%%
where $C_i^j$ is a constant that is equivalent to zero in this case, $U_i(C_i+C_i^j+r_i^j)$ is the utility function of the summation of the rate $C_i$ allocated to the application running on the $i^{th}$ user by carriers $\{1,2,...,j-1\}$ and the rate $r_i^{j,\text{all}}$ allocated to the same application by carrier $j$ where $r_i^{j,\text{all}}=C_i^j+r_i^j$, $\textbf{r}^j =\{r_1^j,r_2^j,...,r_{M_j}^j\}$ and $M_j$ is the number of users in $\mathcal{M}_j$ (i.e. both VIP and regular users) located under the coverage area of the $j^{th}$ carrier. After the eNodeB performs the RA process for the $j^{th}$ carrier by solving optimization problem (\ref{eqn:opt_case1}), the total rate allocated to each user by the eNodeB is equivalent to $C_i+r_i^{j,\text{all}}$. In optimization problem (\ref{eqn:opt_case1}), we consider a utility proportional fairness objective function, based on carrier aggregation, that the eNodeB seeks to maximize when it performs RA for carrier $j$.\\
%%%%%%%%%%%%%%%%%%%%%%%%%%%%%%%%%%%%%%%
%%%%%%%%%%%%%%%%%%%%%%%%%%%%%%%%%%%%%%%
\\\textit{Case$\;2.$ RA Optimization Problem when $ C_i < r_i^{\text{req}}$ for any user $i \in \mathcal{M}_j$ and $\sum_{i=1}^{M_j^{\text{VIP}}} q_i^j \geq R_j$ where $q_i^j=0$ if $C_i \geq r_i^{\text{req}}$ and $q_i^j=r_i^{\text{req}}-C_i$ if $C_i < r_i^{\text{req}}$}:

The eNodeB selects the optimization problem of this case to allocate the $j^{th}$ carrier resources if the total aggregated rate $C_i$ for any user $i$ is less than the user's application minimum required rate $r_i^{\text{req}}$ and $\sum_{i=1}^{M_j^{\text{VIP}}}q_i^j$ for VIP users in $\mathcal{M}_j^{\text{VIP}}$ is greater than or equal the carrier's available resources $R_j$. In this case, the eNodeB allocates the $j^{th}$ carrier resources only to VIP UEs in $\mathcal{M}_j^{\text{VIP}}$ as they are considered more important and regular users in $\mathcal{M}_j^{\text{Reg}}$ are not allocated any of the $j^{th}$ carrier resources since the carrier's resources are limited. The RA optimization problem for the $j^{th}$ carrier in this case is given by:

\begin{equation}\label{eqn:opt_case2}
\begin{aligned}
& \underset{\textbf{r}^j}{\text{max}}
& & \prod_{i=1}^{M_j^{\text{VIP}}}U_i(C_i+C_i^j+r_i^j) \\
& \text{subject to}
& & \sum_{i=1}^{M_j^{\text{VIP}}}r_i^{j,\text{all}} \leq R_j,\;\;r_i^{j,\text{all}} \geq 0\\
& & & C_i=\sum_{l=1}^{j-1}r_i^{l,\text{all}},\;\;r_i^{j,\text{all}}=r_i^j+C_i^j\\
& & & C_i^j=0\\
& & & q_i^j=
\begin{cases}
	0 \;\;\;\;\;\;\;\;\;\;\;\;\; \text{if}\;\; C_i \geq r_i^{\text{req}}\\
	r_i^{\text{req}}-C_i \;\;\text{if}\;\; C_i < r_i^{\text{req}}
\end{cases}\\
& & & \sum_{i=1}^{M_j^{\text{VIP}}}q_i^j \geq R_j, \;\;\; i = 1,2, ...,M_j^{\text{VIP}},\\
\end{aligned}
\end{equation}
%%%%%%%%%%%%%%%%%%%%%%%%5
where $\textbf{r}^j =\{r_1^j,r_2^j,...,r_{M_j^{\text{VIP}}}^j\}$, $C_i^j=0$ and $M_j^{\text{VIP}}$ is the number of users in $\mathcal{M}_j^{\text{VIP}}$. After the eNodeB performs the RA process for the $j^{th}$ carrier by solving optimization problem (\ref{eqn:opt_case2}), each VIP user in $\mathcal{M}_j^{\text{VIP}}$ is allocated a rate that is equivalent to $r_i^{j,\text{all}}$ by carrier $j$ whereas users in $\mathcal{M}_j^{\text{Reg}}$ are not allocated any of the $j^{th}$ carrier resources. The total rate allocated by the eNodeB to each user is equivalent to $C_i+r_i^{j,\text{all}}$. In optimization problem (\ref{eqn:opt_case2}), we consider a utility proportional fairness objective function, based on carrier aggregation, that the eNodeB seeks to maximize when it performs RA for carrier $j$.\\
%%%%%%%%%%%%%%%%%%%%%%%%%%%%%%%%%%%%%%%
%%%%%%%%%%%%%%%%%%%%%%%%%%%%%%%%%%%%%%%
\\\textit{Case$\;3.$ RA Optimization Problem when $C_i < r_i^{\text{req}}$ for any user $i \in \mathcal{M}_j^{\text{VIP}}$ and $\sum_{i=1}^{M_j^{\text{VIP}}} q_i^j < R_j$ where $q_i^j=0$ if $C_i \geq r_i^{\text{req}}$ and $q_i^j=r_i^{\text{req}}-C_i$ if $C_i < r_i^{\text{req}}$}:

The eNodeB selects the optimization problem of this case to allocate the $j^{th}$ carrier resources if the total aggregated rate $C_i$ for any user $i$ is less than the user's application minimum required rate $r_i^{\text{req}}$ and the summation $\sum_{i=1}^{M_j^{\text{VIP}}}q_i^j$ for VIP users in $\mathcal{M}_j^{\text{VIP}}$ is less than the carrier's available resources $R_j$. In this case, the eNodeB allocates the $j^{th}$ carrier resources to all UEs in $\mathcal{M}_j$. The RA optimization problem for the $j^{th}$ carrier in this case is given by:

\begin{equation}\label{eqn:opt_case3}
\begin{aligned}
& \underset{\textbf{r}^j}{\text{max}}
& & \prod_{i=1}^{M_j}U_i(C_i+C_i^j+r_i^j) \\
& \text{subject to}
& & \sum_{i=1}^{M_j}r_i^{j,\text{all}} \leq R_j,\;\;r_i^{j,\text{all}} \geq 0\\
& & & C_i=\sum_{l=1}^{j-1}r_i^{l,\text{all}},\;\;r_i^{j,\text{all}}=r_i^j+C_i^j\\
& & & C_i^j=
\begin{cases}
	0 \;\;\;\;\;\;\;\;\;\;\;\;\; \text{if}\;\; C_i \geq r_i^{\text{req}}\\
	r_i^{\text{req}}-C_i \;\;\text{if}\;\; C_i < r_i^{\text{req}}
\end{cases}\\
& & & q_i^j=
\begin{cases}
	0 \;\;\;\;\;\;\;\;\;\;\;\;\; \text{if}\;\; C_i \geq r_i^{\text{req}}\\
	r_i^{\text{req}}-C_i \;\;\text{if}\;\; C_i < r_i^{\text{req}}
\end{cases}\\
& & & \sum_{i=1}^{M_j^{\text{VIP}}}q_i^j < R_j, \;\; i = 1,2, ...,M_j,\\
\end{aligned}
\end{equation}
%%%%%%%%%%%%%%%%%%%%%%%%%%%%%%%
where $\textbf{r}^j =\{r_1^j,r_2^j,...,r_{M_j}^j\}$ and $M_j$ is the number of users in $\mathcal{M}_j$. After the eNodeB performs the RA process for the $j^{th}$ carrier by solving optimization problem (\ref{eqn:opt_case3}), each user in $\mathcal{M}_j$ is allocated a rate that is equivalent to $r_i^{j,\text{all}}$ by carrier $j$ and the total rate allocated by the eNodeB to each user is equivalent to $C_i+r_i^{j,\text{all}}$. In optimization problem (\ref{eqn:opt_case3}), we consider a utility proportional fairness objective function, based on carrier aggregation, that the eNodeB seeks to maximize when it performs RA for carrier $j$.

Each of the three RA optimization problems (\ref{eqn:opt_case1}), (\ref{eqn:opt_case2}) and (\ref{eqn:opt_case3}) of the $j^{th}$ carrier can be expressed by the following generalized optimization problem:
\begin{equation}\label{eqn:opt_General}
\begin{aligned}
& \underset{\textbf{r}^j}{\text{max}}
& & \prod_{i=1}^{|\alpha_j|}U_i(C_i+C_i^j+r_i^j) \\
& \text{subject to}
& & \sum_{i=1}^{|\alpha_j|}r_i^{j,\text{all}} \leq R_j,\;\;r_i^{j,\text{all}} \geq 0\\
& & & C_i=\sum_{l=1}^{j-1}r_i^{l,\text{all}},\;\;r_i^{j,\text{all}}=r_i^j+C_i^j\\
& & & q_i^j=
\begin{cases}
	0 \;\;\;\;\;\;\;\;\;\;\;\;\; \text{if}\;\; C_i \geq r_i^{\text{req}}\\
	r_i^{\text{req}}-C_i \;\;\text{if}\;\; C_i < r_i^{\text{req}}
\end{cases}\\
& & & i = 1,2, ...,|\alpha_j|,\\
\end{aligned}
\end{equation}
where $C_i^j$ and $\alpha_j$ in (\ref{eqn:opt_General}) are given by
\begin{align*}
C_i^j=
\begin{cases}
	0 \;\;\;\;\;\;\;\;\;\;\;\;\; \text{if}\;\; C_i \geq r_i^{\text{req}}\\
	r_i^{\text{req}}-C_i \;\;\text{if}\;\; C_i < r_i^{\text{req}}\;\;\text{and}\;\;\sum_{i=1}^{|\mathcal{M}_j^{\text{VIP}}|}q_i^j < R_j\\
    0 \;\;\;\;\;\;\;\;\;\;\;\;\; \text{if}\;\; C_i < r_i^{\text{req}}\;\;\text{and}\;\;\sum_{i=1}^{|\mathcal{M}_j^{\text{VIP}}|}q_i^j \geq R_j\\
\end{cases}\\
\end{align*}

\begin{equation}\label{eqn:alpha}
\begin{aligned}
\alpha_j=
\begin{cases}
	\mathcal{M}_j \;\; & \text{if}\;\; C_i \geq r_i^{\text{req}}\:\:\forall i \in \mathcal{M}_j\\
	\mathcal{M}_j^{\text{VIP}}\;\;\; & \text{if}\;\; C_i < r_i^{\text{req}} \; \text{for any user} \; i \in \mathcal{M}_j \\
    & \; \text{and} \; \sum_{i=1}^{M_j^{\text{VIP}}} q_i^j \geq R_j\\
    \mathcal{M}_j \;\; & \text{if}\;\; C_i < r_i^{\text{req}} \; \text{for any user} \; i \in \mathcal{M}_j^{\text{VIP}}\\
    & \; \text{and} \; \sum_{i=1}^{M_j^{\text{VIP}}} q_i^j < R_j
\end{cases}
\end{aligned}
\end{equation}
%%%%%%%%%%%%%%%%%%%%%%%%%%%%%%%
where $\textbf{r}^j =\{r_1^j,r_2^j,...,r_{|\alpha_j|}^j\}$, $\alpha_j$ is a set of users located under the coverage area of carrier $j$ that is equivalent to $\mathcal{M}_j$ or $\mathcal{M}_j^{\text{VIP}}$ based on certain conditions as shown in (\ref{eqn:alpha}) and $|\alpha_j|$ is the number of users in $\alpha_j$.

The objective function in optimization problem (\ref{eqn:opt_General}) is equivalent to $\sum_{i=1}^{|\alpha_j|}\log U_i(C_i+C_i^j+r_i^j)$. Later in this section we prove that optimization problem (\ref{eqn:opt_General}) is a convex optimization problem and there exists a unique tractable global optimal solution. Once the eNodeB is done performing the RA process, for the $j^{th}$ carrier, by solving optimization problem (\ref{eqn:opt_General}), each user in $\alpha_j$ is allocated a rate that is equivalent to $r_i^{j,\text{all}}=r_i^j+C_i^j$ and the user's total aggregated rate allocated by the eNodeB from carriers $\{1,2,...,j\}$ is given by $\sum_{l=1}^{j}r_i^{l,\text{all}}$.

\begin{lem}\label{lem:concavity}
The utility functions $\log U_i(C_i+C_i^j+r_i^j)$ in optimization problem (\ref{eqn:opt_General}) are strictly concave functions.
\end{lem}

\begin{proof}
The utility functions are assumed to be logarithmic or sigmoidal-like functions as discussed in Section \ref{sec:utility functions}. Therefore, $U_i(C_i+C_i^j+r_i^j)$ is a strictly concave (i.e. in the case of logarithmic utility functions) or a sigmoidal-like function of the total aggregated rate $C_i+C_i^j+r_i^j$ allocated to user $i$ application from carriers $\{1,2,...,j\}$ after performing the RA process of the $j^{th}$ carrier by the eNodeB.

In the case of logarithmic utility function, recall the utility function properties in Section \ref{sec:utility functions}, the utility function of the application rate is positive, increasing and twice differentiable with respect to the application rate. It follows that $U'_i(C_i+C_i^j+r_i^j) = \frac{dU_i(C_i+C_i^j+r_i^j)}{dr_i^j} > 0$ and $U''_i(C_i+C_i^j+r_i^j) = \frac{d^2U_i(C_i+C_i^j+r_i^j)}{{dr_i^j}^2} < 0$, i.e. since $C_i+C_i^j$ is greater or equal zero. Then the function $\log U_i(C_i+C_i^j+r_i^j)$ has $\frac{d \log (U_i(C_i+C_i^j+r_i^j))}{dr_i^j}=\frac{U_i'(C_i+C_i^j+r_i^j)}{U_i(C_i+C_i^j+r_i^j)} > 0$ and $\frac{d^2 \log (U_i(C_i+C_i^j+r_i^j))}{{dr_i^j}^2}=\frac{U''_i(C_i+C_i^j+r_i^j)U_i(C_i+C_i^j+r_i^j)-U'^2_i(C_i+C_i^j+r_i^j)}{U^2_i(C_i+C_i^j+r_i^j)} < 0$. Therefore, the natural logarithm of the logarithmic utility function $\log(U_i(C_i+C_i^j+r_i^j))$ is strictly concave.

On the other hand, in the case of sigmoidal-like utility function, the normalized sigmoidal-like function is given by $U_i(C_i+C_i^j+r_i^j)=c_i\Big(\frac{1}{1+e^{-a_i(C_i+C_i^j+r_i^j-b_i)}}-d_i\Big)$. For $0 < r_i^j < (R_j-C_i^j)$, we have

\begin{equation*}\label{eqn:sigmoid_bound}
\begin{aligned}
0&<c_i\Big(\frac{1}{1+e^{-a_i(C_i+C_i^j+r_i^j-b_i)}}-d_i\Big)<1\\
d_i&<\frac{1}{1+e^{-a_i(C_i+C_i^j+r_i^j-b_i)}}<\frac{1+c_id_i}{c_i}\\
\frac{1}{d_i}&>{1+e^{-a_i(C_i+C_i^j+r_i^j-b_i)}}>\frac{c_i}{1+c_id_i}\\
0&<1-d_i({1+e^{-a_i(C_i+C_i^j+r_i^j-b_i)}})<\frac{1}{1+c_id_i}\\
\end{aligned}
\end{equation*}

It follows that for $0 < r_i^j < (R_j-C_i^j)$, we have the first and second derivatives as
\begin{align*}
\frac{d}{dr_i^j}\log U_i(C_i+&C_i^j+r_i^j) =\\
& \frac{a_id_i e^{-a_i(C_i+C_i^j+r_i^j-b_i)}}{1-d_i(1+e^{-a_i(C_i+C_i^j+r_i^j-b_i)})} \\
&  + \frac{a_ie^{-a_i(C_i+C_i^j+r_i^j-b_i)}}{(1+e^{-a_i(C_i+C_i^j+r_i^j-b_i)})}>0\\
\frac{d^2}{{dr_i^j}^2}\log U_i(C_i+&C_i^j+r_i^j) =\\
& \frac{-a_i^2d_ie^{-a_i(C_i+C_i^j+r_i^j-b_i)}}{c_i\Big(1-d_i(1+e^{-a(C_i+C_i^j+r_i^j-b_i)})\Big)^2} \\
&  + \frac{-a_i^2e^{-a_i(C_i+C_i^j+r_i^j-b_i)}}{(1+e^{-a_i(C_i+C_i^j+r_i^j-b_i)})^2} < 0.\\
\end{align*}
Therefore, the natural logarithm of the sigmoidal-like utility function $\log (U_i(C_i+C_i^j+r_i^j)$ is strictly concave function. Therefore, the utility functions natural logarithms have strictly concave natural logarithms in both cases of logarithmic utility functions and sigmoidal-like utility functions.
\end{proof}

%%%%%%%%%%%%%%%%%%%%%%%%%%%%%%%%
Theorem \ref{thm:global_soln} proves the convexity of optimization problem (\ref{eqn:opt_General}).

\begin{thm}\label{thm:global_soln}
Optimization problem (\ref{eqn:opt_General}) is a convex optimization problem and there exists a unique tractable global optimal solution.
\end{thm}

\begin{proof}
It follows from Lemma \ref{lem:concavity} that all UEs utility functions of applications rates are strictly concave. Therefore, optimization problem (\ref{eqn:opt_General}) is a convex optimization problem. For a convex optimization problem there exists a unique tractable global optimal solution \cite{Boyd}.
\end{proof}

\section{RA Optimization Algorithm}\label{sec:Algorithm}

In this section, we present our multi-carrier resource allocation with user discrimination algorithm. The proposed algorithm consists of UE and eNodeB parts shown in Algorithm \ref{alg:UE} and Algorithm \ref{alg:eNodeB}, respectively. The execution of the algorithm starts by UEs, subscribing for mobile services, transmitting their application utility parameters to the eNodeB, which allocates available carriers' resources to UEs based on a proportional fairness policy. First, the eNodeB performs the user grouping method described in Section \ref{sec:UsersGrouping} for each carrier by creating three user group sets $\mathcal{M}_j^{\text{VIP}}$, $\mathcal{M}_j^{\text{Reg}}$ and $\mathcal{M}_j$ for UEs located within the coverage area of the $j^{th}$ carrier. It then starts performing the RA process to allocate the carriers resources starting with carrier $1$ in $\mathcal{K}$ (i.e. the carrier with the smallest coverage radius) in ascending order $1 \rightarrow K$. In order to allocate certain carrier's resources, the eNodeB performs the RA process that corresponds to the most appropriate case among the three cases presented in Section \ref{sec:ResourceAllocation}. From optimization problem (\ref{eqn:opt_General}), we have the following Lagrangian

\begin{equation}\label{eqn:lagrangian}
\begin{aligned}
L(\textbf{r}^j,p^j) =
& \sum_{i=1}^{|\alpha_j|}\log U_i(C_i+C_i^j+r_i^j)\\
& -p^j(\sum_{i=1}^{|\alpha_j|}(C_i^j+r_i^j) + \sum_{i=1}^{|\alpha_j|}z_{i} - R_j),\\
\end{aligned}
\end{equation}

where $z_{i}\geq 0$ is the slack variable and $p^j$ is Lagrange multiplier that represents the shadow price (price per unit bandwidth for all the $|\alpha_j|$ channels). The rates, solutions to equation (\ref{eqn:opt_General}), are the values $r_i^j$ which solve equation $\frac{\partial \log U_i(C_i+C_i^j+r_i^j)}{\partial r_i^j} = p^j$ and are the intersection of the time varying shadow price, horizontal line $y = p^j$, with the curve $y = \frac{\partial \log U_i(C_i+C_i^j+r_i^j)}{\partial r_i^j}$ geometrically. The rate allocated by carrier $j$ to the $i^{th}$ UE is equivalent to $r_i^{j,\text{all}}=r_i^j+C_i^j$. When the eNodeB is done allocating the $K^{th}$ carrier resources, each user is then allocated its final aggregated rate $r_i = \sum_{j=1}^{K}r_i^{j,\text{all}}$.

\begin{algorithm}
\caption{The $i^{th}$ UE Algorithm}\label{alg:UE}
\begin{algorithmic}
\LOOP
      \STATE {Send application utility parameters $k_i$, $a_i$, $b_i$, $r_i^{\text{max}}$ and $r_i^{\text{req}}$ to eNodeB.}
      \STATE {Receive the final allocated rate $r_i$ from the eNodeB.}
\ENDLOOP
\end{algorithmic}
\end{algorithm}

%%%%eNodeB pseudocode
\begin{algorithm}
\caption{The eNodeB Algorithm}\label{alg:eNodeB}
\begin{algorithmic}
\LOOP
\STATE {Initialize $C_i=0$; $C_i^j=0$; $r_i^{j,\text{all}}=0$.}
\STATE {Receive application utility parameters $k_i$, $a_i$, $b_i$, $r_i^{\text{max}}$ and $r_i^{\text{req}}$ from all UEs in $\mathcal{M}$.}%
\FOR {$j \leftarrow 1$  to  $K$}
 \STATE{Create user groups $\mathcal{M}_j^{\text{VIP}}$, $\mathcal{M}_j^{\text{Reg}}$ and $\mathcal{M}_j$ for UEs located within the coverage area of the $j^{th}$ carrier.}
\ENDFOR
\FOR {$i \leftarrow 1$  to  $|\mathcal{M}_j|$}
 \STATE{Create carrier group $\mathcal{K}_i$ for the $i^{th}$ UE's all in range carriers.}
\ENDFOR
\FOR {$j \leftarrow 1$  to  $K$}
 \IF{$C_i < r_i^{\text{req}}$}
 \STATE{$q_i^j=r_i^{\text{req}}-C_i$}
 \ELSE
 \STATE{$q_i^j=0$}
 \ENDIF
 \IF{$ C_i \geq r_i^{\text{req}}\:\:\forall i \in \mathcal{M}_j$}
 \STATE{$C_i^j=0$}
 \STATE{Solve $\textbf{r}^j =  \arg \underset{\textbf{r}^j}\max \sum_{i=1}^{|\mathcal{M}_j|}\log U_i(C_i+C_i^j+r_i^j) - p^j(\sum_{i=1}^{|\mathcal{M}_j|}(r_i^j+C_i^j)-R_j)$.}
 \STATE{Allocate rate $r_i^{j,\text{all}}=r_i^j+C_i^j$ by the $j^{th}$ carrier to each user in $\mathcal{M}_j$.}
 \STATE{Calculate new $C_i=C_i+r_i^{j,\text{all}}\:\:\forall i \in \mathcal{M}_j$}
 %\ENDIF
  \ELSIF{$C_i < r_i^{\text{req}}$ for any user $i \in \mathcal{M}_j$ $\&\&$ $\sum_{i=1}^{M_j^{\text{VIP}}} q_i^j \geq R_j$}
  %$\sum_{i=1}^{|\mathcal{M}_j^{\text{VIP}}|}(r_i^{\text{req}})>R_j$ $\&\&$ $ C_i < r_i^{\text{req}}$ for any user $i \in \mathcal{M}_j$}
 \STATE{$C_i^j=0$}
 \STATE{Solve $\textbf{r}^j =  \arg \underset{\textbf{r}^j}\max \sum_{i=1}^{|\mathcal{M}_j^{\text{VIP}}|}\log U_i(C_i+C_i^j+r_i^j) - p^j(\sum_{i=1}^{|\mathcal{M}_j^{\text{VIP}}|}(r_i^j+C_i^j)-R_j)$.}
 \STATE{Allocate rate $r_i^{j,\text{all}}=r_i^j+C_i^j$ by the $j^{th}$ carrier to each user in $\mathcal{M}_j^{\text{VIP}}$.}
 \STATE{Calculate new $C_i=C_i+r_i^{j,\text{all}}\:\:\forall i \in \mathcal{M}_j^{\text{VIP}}$}
 \ELSIF{$C_i < r_i^{\text{req}}$ for any user $i \in \mathcal{M}_j^{\text{VIP}}$ and $\sum_{i=1}^{|\mathcal{M}_j^{\text{VIP}}|}q_i^j < R_j$}
 \IF{$C_i < r_i^{\text{req}}$}
 \STATE{$C_i^j=r_i^{\text{req}}-C_i$}
 \ELSE
 \STATE{$C_i^j=0$}
 \ENDIF
 \STATE{Solve $\textbf{r}^j =  \arg \underset{\textbf{r}^j}\max \sum_{i=1}^{|\mathcal{M}_j|}\log U_i(C_i+C_i^j+r_i^j) - p^j(\sum_{i=1}^{|\mathcal{M}_j|}(r_i^j+C_i^j)-R_j)$.}
 \STATE{Allocate rate $r_i^{j,\text{all}}=r_i^j+C_i^j$ by the $j^{th}$ carrier to each user in $\mathcal{M}_j$.}
 \STATE{Calculate new $C_i=C_i+r_i^{j,\text{all}}\:\:\forall i \in \mathcal{M}_j$}
 \ENDIF
\ENDFOR
\STATE{Allocate total aggregated rate $r_i = \sum_{j=1}^{K}r_i^{j,\text{all}}$ by the eNodeB to each UE $i$ in $\mathcal{M}$}
\ENDLOOP
\end{algorithmic}
\end{algorithm}

\section{Simulation Results}\label{sec:sim}
Algorithm \ref{alg:UE} and \ref{alg:eNodeB} were applied in C++ to multiple utility functions with different parameters. Simulation results showed convergence to the global optimal rates. In this section, we consider a mobile cell with one eNodeB, two carriers with available resources and $8$ active UEs located under the coverage area of the eNodeB as shown in Figure \ref{fig:SystemModel}. The UEs are divided into two groups. The $1^{st}$ group of UEs (index $i=\{1,2,3,4\}$) represents user group $\mathcal{M}_1$ located within the coverage radius $D_1$ of carrier $1$. Each user in $\mathcal{M}_1$ belongs to one of the two classes of user groups, i.e. VIP user group and Regular user group, where $\mathcal{M}_1^{\text{VIP}}=\{2,4\}$, $\mathcal{M}_1^{\text{Reg}}=\{1,3\}$ and $\mathcal{M}_1=\mathcal{M}_1^{\text{VIP}}\cup \mathcal{M}_1^{\text{Reg}}$. On the other hand, the $2^{nd}$ group of UEs (index $i=\{1,2,3,4,5,6,7,8\}$) represents user group $\mathcal{M}_2$ located within the coverage radius $D_2$ of carrier $2$. Each user in $\mathcal{M}_2$ belongs to a VIP user group or a regular user group where $\mathcal{M}_2^{\text{VIP}}=\{2,4,6,8\}$, $\mathcal{M}_2^{\text{Reg}}=\{1,3,5,7\}$ and $\mathcal{M}_2=\mathcal{M}_2^{\text{VIP}}\cup \mathcal{M}_2^{\text{Reg}}$.

We use sigmoidal-like utility functions and logarithmic utility functions with different parameters to represent each of the users' applications. We use three normalized sigmoidal-like functions that are expressed by equation (\ref{eqn:sigmoid}) with different parameters. The used parameters are $a_i = 5$, $b_i=10$ that correspond to a sigmoidal-like function with inflection point $r_i =10$ which represents the utility of UE with index $i=\{5\}$, $a_i = 3$, $b_i=20$ that correspond to a sigmoidal-like function with inflection point $r_i=20$ which represents the utility of UE with index $i=\{1\}$, and $a_i = 1$,  $b_i=30$ that correspond to a sigmoidal-like function with inflection point $r_i=30$ which represents the utility of UEs with indexes $i=\{2,6\}$, as shown in Figure \ref{fig:utility}. We use three logarithmic functions expressed by equation (\ref{eqn:log}) with $r_i^{\text{max}} =100$ and different $k_i$ parameters to represent delay-tolerant applications. We use $k_i =15$ for UE with index $i=\{7\}$, $k_i =3$ for UE with index $i=\{3\}$, and $k_i = 0.5$ for UEs with indexes $i=\{4,8\}$, as shown in Figure \ref{fig:utility}. A summary is shown in table \ref{table:parameters}. We use an application minimum required rate that is equivalent to the inflection point of the sigmoidal-like function, i.e. $r_i^{\text{req}}=b_i$, for each VIP user running a real-time application, we use $r_i^{\text{req}}=15$ for each VIP user running a delay-tolerant application and $r_i^{\text{req}}=0$ for each regular user whether it is running real-time application or delay-tolerant application.
%%%%%%%%%%%%%%%%%%%%%%%%%%
\begin{figure}[tb]
\centering
\includegraphics[height=2.5in, width=2.5in]{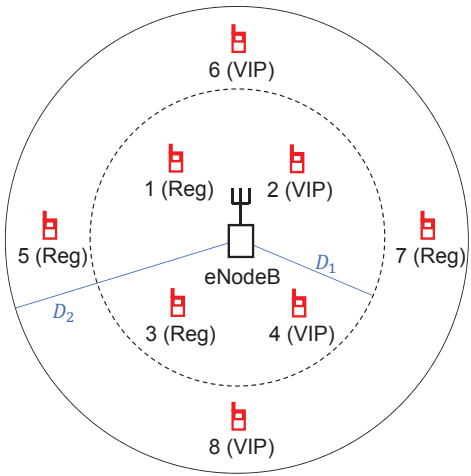}
\caption{System model for a mobile system with $M=8$ users and $K=2$ carriers available at the eNodeB. Carrier $1$ coverage radius is $D_1$ and carrier $2$ coverage radius is $D_2$ with $D_1 < D_2$. $\mathcal{M}_1=\{1,2,3,4\}$ and $\mathcal{M}_2=\{1,2,...,8\}$ represent the sets of user groups located under the coverage area of carrier $1$ and carrier $2$, respectively.}
%%\myfigureshrinker{\vspace{-0.06in}}
\label{fig:SystemModel}
\end{figure}
%%%%%%%%%%%%%%%%%%%%%%%%%
\begin{figure}[tb]
\centering
\includegraphics[width=3.5in]{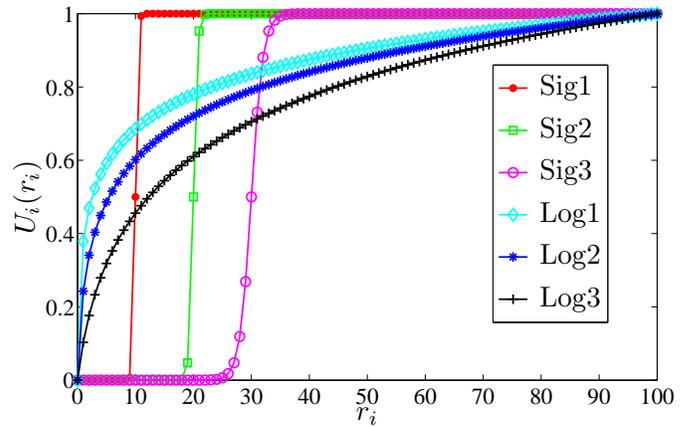}
\caption{The users utility functions $U_i(r_i)$ used in the simulation (three sigmoidal-like functions and three logarithmic functions).}
%%\myfigureshrinker{\vspace{-0.06in}}
\label{fig:utility}
\end{figure}
%%%%%%%%%%%%%%%%%%%%%%%%%%
\begin {table}[]
\caption {Users and their applications utilities}
\label{table:parameters}
\begin{center}
\renewcommand{\arraystretch}{1.4} %<- modify value to suit your needs
\begin{tabular}{| l | l | l | }
%\label{table:utility}
  \hline
  \multicolumn{2}{|c|}{Applications Utilities Parameters} & \multicolumn{1}{|c|}{Users Indexes} \\  \hline
  Sig1 & Sig $a_i=5,\:\: b_i=10$  &  $i=\{5\}$ \\ \hline
  Sig2 & Sig $a_i=3,\:\: b_i=20$ & $i=\{1\}$  \\ \hline
  Sig3 & Sig $a_i=1,\:\: b_i=30$ & $i=\{2,6\}$   \\ \hline
  Log1 & Log $k_i=15,\:\: r_i^{\text{max}}=100$ & $i=\{7\}$   \\ \hline
  Log2 & Log $k_i=3,\:\: r_i^{\text{max}}=100$ & $i=\{3\}$   \\ \hline
  Log3 & Log $k_i=0.5,\:\: r_i^{\text{max}}=100$ & $i=\{4,8\}$ \\ \hline
\end{tabular}
%\caption {Should be a caption}
\end{center}
\end {table}
%%%%%%%%%%%%%%%%%%%%%%%%%%%%%%%%%%%%%%%%%%%%%%
\subsection{Carrier $1$ Allocated Rates for $60\le R_1 \le 150$}
In the following simulations, we set $\delta =10^{-3}$, carrier $1$ rate $R_1$ takes values between $60$ and $150$ with step of $10$. In Figure \ref{fig:ri_carrier1_all_R1}, we show the allocated rates $r_i^{1,\text{all}}$ of different users with different values of carrier $1$ total rate $R_1$ and observe how the proposed rate allocation algorithm converges for different values of $R_1$. In Figure \ref{fig:ri_carrier1_all_R1}, we show that both VIP and regular users in user group $\mathcal{M}_1$ are allocated resources by carrier $1$ when $60\le R_1 \le 150$ since carrier $1$ available resources $R_1$ is greater than the total applications minimum required rates for users in $\mathcal{M}_1$. Figure \ref{fig:ri_carrier1_all_R1} also shows that by using the proposed RA with user discrimination algorithm, no user is allocated zero rate (i.e. no user is dropped). However, carrier $1$ resources are first allocated to the VIP users until each of their applications reaches the application minimum required rate $r_i^{\text{req}}$. Then the majority of carrier $1$ resources are allocated to the UEs running adaptive real-time applications until they reach their inflection rates, the eNodeB then allocates more of carrier $1$ resources to UEs with delay-tolerant applications.
%%%%%%%%%%%%%%%%%%%%%%%%%%%%%%%%%%%%%%%%%%%%%%%%%
\begin{figure}[tb]
\centering
\includegraphics[height=2in, width=3.5in]{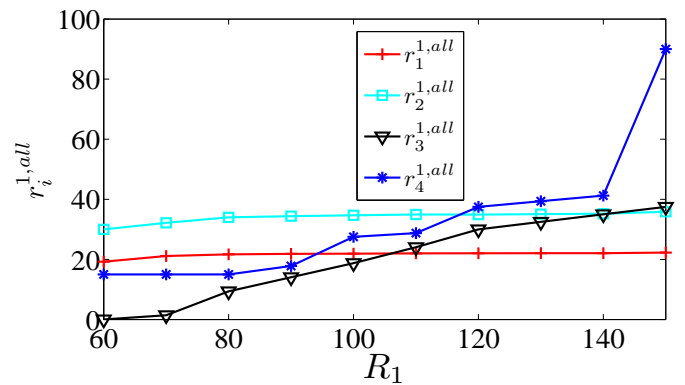}
\caption{The rates $r_i^{1,\text{all}}$ allocated from carrier $1$ to $\mathcal{M}_1$ user group with carrier $1$ available resources $60<R_1<150$.}
%%\myfigureshrinker{\vspace{-0.06in}}
\label{fig:ri_carrier1_all_R1}
\end{figure}
%%%%%%%%%%%%%%%%%%%%%%%%%%%%%%%%%%%%%%%%%%%%%%%%%%%%%%%%
\subsection{Carrier $2$ Allocated Rates and the Total Aggregated Rates for $10\le R_2 \le 150$}
In the following simulations, we set $\delta =10^{-3}$, carrier $2$ rate $R_2$ takes values between $10$ and $150$ with step of $10$ and carrier $1$ rate is fixed at $R_1 = 60$. In Figure \ref{fig:ri_carrier2}, we show the allocated rates $r_i^{2,\text{all}}$ and the final aggregated rates $r_i$ of different users with different values of carrier $2$ total rate $R_2$ and observe how the proposed rate allocation algorithm converges for different values of $R_2$. In Figure \ref{fig:ri_carrier2_all}, we show that when $10\le R_2 \le 45$ only VIP users in $\mathcal{M}_2$ (i.e. UEs in $\mathcal{M}_2^{\text{VIP}}$) that were not allocated resources by carrier $1$ or did not reach their applications minimum required rates are allocated resources by carrier $2$. Whereas when $45 < R_2 \le 150$, both VIP and regular users in $\mathcal{M}_2$ are allocated resources by carrier $2$ as carrier $2$ total rate $R_2$ is greater than $\sum_{i=1}^{M_2^{\text{VIP}}} q_i^2$ (i.e. the total required rates for UEs to reach their $r_i^{\text{req}}$). Figure \ref{fig:ri_carrier2_all} also shows that by using the proposed RA with user discrimination algorithm that is based on carrier aggregation, the eNodeB takes into consideration the rates allocated to users in $\mathcal{M}_2$ by carrier $1$ when allocating carrier $2$ resources. Carrier $2$ resources are first allocated to VIP users until each of their applications reaches the application minimum required rate $r_i^{\text{req}}$. Then the majority of carrier $2$ resources are allocated to the UEs running adaptive real-time applications until they reach their inflection rates, the eNodeB then allocates more of carrier $2$ resources to UEs with delay-tolerant applications.

Figure \ref{fig:ri_Aggregated} shows the total aggregated rates $r_i = \sum_{j=1}^{2}r_i^{j,\text{all}}$ for the $8$ users.

\begin{figure}[tb]
  \centering
  \subfigure[The rates $r_i^{2,\text{all}}$ allocated from carrier $2$ to $\mathcal{M}_2$ user group.]{%
  \label{fig:ri_carrier2_all}
  \includegraphics[width=3.5in]{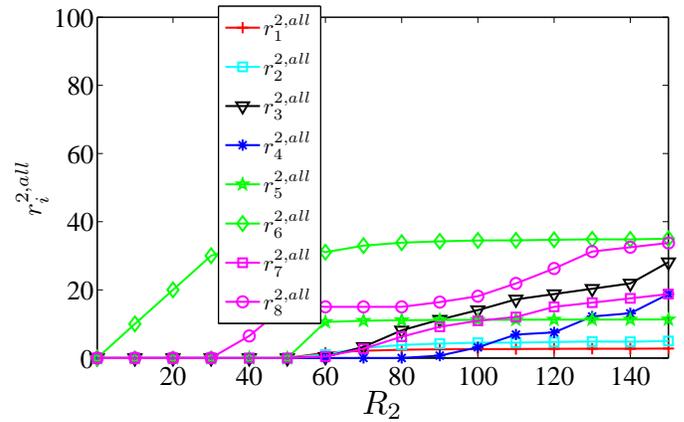}
  }\\%
\subfigure[The total aggregated rates $r_i$ allocated by the eNodeB to the $8$ users.]{%
  \label{fig:ri_Aggregated}
  \includegraphics[width=3.5in]{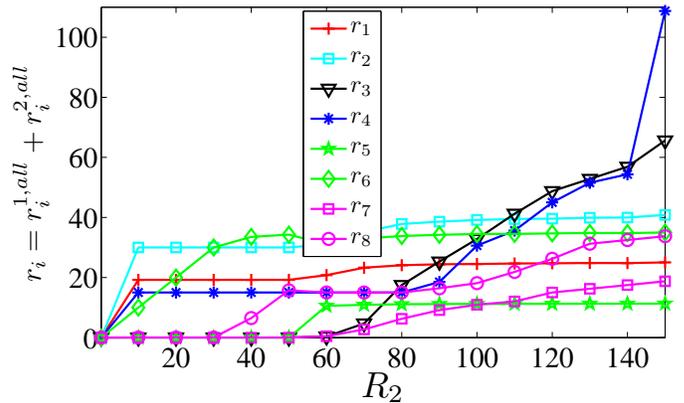}
  }%
\caption{The rates $r_i^{2,\text{all}}$ allocated from carrier $2$ to users in $\mathcal{M}_2$ and the total aggregated rates allocated to the $8$ users with carrier $2$ available resources $10<R_2<150$ and carrier $1$ resources fixed at $R_1=60$.}
\label{fig:ri_carrier2}
\end{figure}
%%%%%%%%%%%%%%%%%%%%%%%%%%%%%%%%
\subsection{Pricing Analysis for Carrier $1$ and Carrier $2$}
In the following simulations, we set $\delta =10^{-3}$. In Figure \ref{fig:price_carrier1}, we show carrier $1$ shadow price with $60\leq R_1 \leq 150$. We observe that carrier $1$ price $p^1$ is traffic-dependant as it decreases for higher values of $R_1$. In Figure \ref{fig:price_carrier1_carrier2}, we show the offered price of carrier $2$ with $10\leq R_2 \leq 150$ and $R_1=60$. We observe that $p^2$ decreases when $R_2$ increases for $10\le R_2 \le 45$, only VIP users are allocated rates by carrier $2$ when $10\le R_2 \le 45$. However, we observe a jump in the price when $R_2=50$ as more users are considered in the rate allocation process (i.e VIP users and regular users in $\mathcal{M}_2$). Figure \ref{fig:price_carrier1_carrier2} also shows that carrier $2$ price $p^2$ decreases when $R_2$ increases for $50\le R_2 \le 150$.
%%%%%%%%%%%%%%%%%%%%%%%%%%%%%%%
\begin{figure}
\centering
\includegraphics[height=2in, width=3.5in]{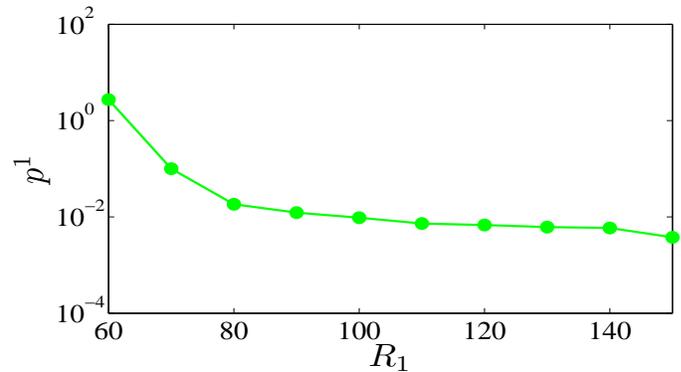}
\caption{Carrier $1$ shadow price $p^1$ with carrier $1$ resources $60<R_1<150$.}
%%\myfigureshrinker{\vspace{-0.06in}}
\label{fig:price_carrier1}
\end{figure}

\begin{figure}[tb]
\centering
\includegraphics[height=2in, width=3.5in]{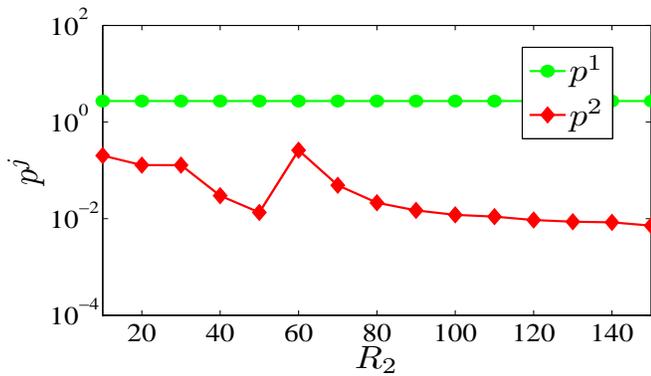}
\caption{Carrier $1$ shadow price $p^1$ and carrier $2$ shadow price $p^2$ with carrier $2$ resources $10<R_2<150$ and carrier $1$ resources fixed at $R_1=60$.}
%%\myfigureshrinker{\vspace{-0.06in}}
\label{fig:price_carrier1_carrier2}
\end{figure}
%%%%%%%%%%%%%%%%%%%%%%%%%%%%%%%%%%%
\section{Conclusion}\label{sec:conclude}
In this paper, we proposed an efficient resource allocation with user discrimination approach for 5G systems to allocate multiple carriers resources optimally among UEs that belong to different user groups classes. We used utility functions to represent the applications running on the UEs. Each user is assigned a minimum required application rate based on its class and the type of its application. Users are partitioned into different user groups based on their class and the carriers coverage area. We presented resource allocation optimization problems based on carrier aggregation for different cases. We proved the existence of a tractable global optimal solution. We presented a RA algorithm for allocating resources from different carriers optimally among different classes of mobile users. The proposed algorithm ensures fairness in the utility percentage, gives priority to VIP users and within a VIP or a regular user group it gives priority to adaptive real-time applications while providing a minimum QoS for all users. We showed through simulations that the proposed resource allocation algorithm converges to the optimal rates. We also showed that the pricing provided by our algorithm depends on the traffic load.

\bibliographystyle{ieeetr}
\bibliography{pubs}
\end{document}